%% file: arxiv-version.tex
\documentclass[12pt]{article}

\usepackage{fullpage}
\usepackage{amsmath}
\usepackage{amsthm}
\usepackage{amsfonts}
\usepackage{amssymb}
\usepackage{mathrsfs}
\usepackage{xspace}

\usepackage{hyperref}

\newcommand{\reals}{\mathbb{R}}
\newcommand{\nats}{\mathbb{N}}

\newcommand{\size}{\mathrm{size}}
\newcommand{\Pow}{\mathscr{P}}
\newcommand{\disjointunion}{\mathbin{\dot{\cup}}}

\newcommand{\commentout}[1]{}

\newtheorem{lemma}{Lemma}
\newtheorem{theorem}{Theorem}
\newtheorem{corollary}{Corollary}

\title{{\bf Size-Degree Trade-Offs \\
for Sums-of-Squares \\
and Positivstellensatz Proofs}}

\author{Albert Atserias \qquad
  Tuomas Hakoniemi \\
  Universitat Polit\`ecnica de Catalunya \\
  {\texttt{\{atserias,hakoniemi\}@cs.upc.edu}}
}

\begin{document}

\maketitle

\thispagestyle{empty}

\begin{abstract} 
\input{abstract.tex}
\end{abstract}

\clearpage
\setcounter{page}{1}

\input{section-1-introduction.tex}

\input{section-2-preliminaries.tex}

\input{section-3-degree-reduction-lemma.tex}

\input{section-4-applications.tex}

\input{section-5-duality.tex}

\input{section-6-concluding-remarks.tex}

\bigskip \bigskip \noindent\textbf{Acknowledgments.}  We are grateful
to Michal Garlik, Moritz M\"uller and Aaron Potechin for comments on
an earlier version of this paper. We are also grateful to Jakob
Nordstr\"om for initiating a discussion on the several variants of the
definition of \emph{monomial size} as discussed in
Section~\ref{sec:preliminaries}. Both authors were partially funded by
European Research Council (ERC) under the European Union's Horizon
2020 research and innovation programme, grant agreement ERC-2014-CoG
648276 (AUTAR) and MICCIN grant TIN2016-76573-C2-1P (TASSAT3).

\bibliographystyle{plain}
\bibliography{semialgebraicdegree.bib}

\end{document}

%% file: abstract.tex
We show that if a system of degree-$k$ polynomial constraints on~$n$
Boolean variables has a Sums-of-Squares (SOS) proof of unsatisfiability
with at most~$s$ many monomials, then it also has one whose degree is
of the order of the square root of~$n \log s$ plus~$k$. A similar
statement holds for the more general Positivstellensatz (PS) proofs.
This establishes size-degree trade-offs for SOS and PS that match
their analogues for weaker proof systems such as Resolution,
Polynomial Calculus, and the proof systems for the LP and SDP
hierarchies of Lov\'asz and Schrijver. As a corollary to this, and to
the known degree lower bounds, we get optimal integrality gaps for
exponential size SOS proofs for sparse random instances of the
standard NP-hard constraint optimization problems. We also get
exponential size SOS lower bounds for Tseitin and Knapsack formulas.
The proof of our main result relies on a zero-gap duality theorem for
pre-ordered vector spaces that admit an order unit, whose
specialization to PS and SOS may be of independent interest.

%% file: section-1-introduction.tex
\section{Introduction} \label{sec:introduction}

A key result in semialgebraic geometry is the
Positivstellensatz \cite{Stengle,Krivine}, whose weak form gives a
version of the Nullstellensatz for semialgebraic sets: A system of
polynomial equations $p_1 = 0, \ldots, p_m = 0$ and polynomial
inequalities $q_1 \geq 0, \ldots, q_\ell \geq 0$ on $n$ commuting
variables $x_1,\ldots,x_n$ has no solution over reals if and only if
\begin{equation}
-1 = s_\emptyset + 
\sum_{J\subseteq [\ell] \atop J \not= \emptyset}s_J\prod_{j\in J}q_j + \sum_{j\in [m]}t_j p_j, \label{eqn:generalrefutation}
\end{equation}
where the $s_J$ are sums of squares of polynomials, and the $t_j$ are
arbitrary polynomials. Based on this, Grigoriev and Vorobjov
\cite{GrigorievVorobjov} defined the Positivstellensatz (PS) proof
system for certifying the unsatisfiability of systems of polynomial
inequalities, and initiated the study of its proof complexity.

For most cases of interest, the statement of the Positivstellensatz
stays true even if the first sum in~\eqref{eqn:generalrefutation}
ranges only over singleton sets~\cite{Putinar}. This special case of
PS yields a proof system called Sums-of-Squares (SOS). Starting with
the work in~\cite{BarakBrandaoHarrowKelnerSteurerZhou}, SOS has
received a good deal of attention for its applications
in algorithms and complexity theory. For the former, through the
connection with the hierarchies of SDP
relaxations~\cite{Lasserre,Parrilo,ODonnellZhou,TulsianiChlamtac}. For
the latter, through the lower bounds on the sizes of SDP lifts of
combinatorial
polytopes~\cite{Fiorinietal,LeeRaghavendraSteurerTan,LeeRaghavendraSteurer}.
We refer the reader to the introduction of~\cite{ODonnellZhou} for a
discussion on the history of these proof systems and their relevance
for combinatorial optimization.

In this paper we concentrate on the proof complexity of PS and SOS
when their variables range over the Boolean hypercube, i.e., the
variables come in pairs of twin variables $x_i$ and~$\bar{x}_i$, and
are restricted through the axioms $x_i^2-x_i = 0$, $\bar{x}_i^2 - \bar{x}_i = 0$ and
$x_i+\bar{x}_i-1=0$. This case is most relevant in combinatorial
contexts. It is also the starting point for a direct link with the
traditional proof systems for propositional logic, such as Resolution,
through the realization that monomials represent Boolean disjunctions,
i.e., clauses.
%
%
In return, this link brings concepts and methods from the area of
propositional proof complexity to the study of PS and SOS proofs.

In analogy with the celebrated size-width trade-off for Resolution
\cite{BenSassonWigderson} or the size-degree trade-off for
Polynomial Calculus \cite{ImpagliazzoPudlakSgall}, a question that is
suggested by this link is whether the \emph{monomial size} of a PS
proof can be traded for its \emph{degree}. For a proof as
in~\eqref{eqn:generalrefutation}, the monomial size of the proof is
the number of monomials in an explicit representation
of the summands of the right-hand side. The degree of the proof is the
maximum of the degrees of those summands. These are the two most
natural measures of complexity for PS proofs (and precise definitions
for both these measures will be made in
Section~\ref{sec:preliminaries}). The importance of the question
whether size can be traded for degree stems from the fact that, at the
time of writing, the complexity of PS and SOS proofs is relatively
well understood when it is measured by degree, but rather poorly
understood when it is measured by monomial size. If size could be
traded for degree, then strong lower bounds on degree would transfer
to strong lower bounds on monomial size. The converse, namely that
strong lower bounds on monomial size transfer to strong lower bounds
on degree, has long been known by elementary linear algebra.

In this paper we answer the size-degree trade-off question for SOS,
and for PS proofs of bounded \emph{product width}, i.e., the number of
inequalities that are multiplied together
in~\eqref{eqn:generalrefutation}. We show that if a system of
degree-$k$ polynomial constraints on $n$ pairs of twin variables has
a PS proof of unsatisfiability of product width $w$ and no more than
$s$ many monomials in total, then it also has one of degree $O(\sqrt{n
  \log s} + kw)$. By taking $w = 1$, this yields a size-degree
trade-off for SOS as a special case.

Our result matches its analogues for weaker proof systems that were
considered before. Building on the work of~\cite{BeamePitassi}
and~\cite{CleggEdmondsImpagliazzo}, a size-width trade-off theorem was
established for Resolution: a proof with $s$ many clauses can be
converted into one in which all clauses have size~$O(\sqrt{n\log s}+k)$, 
where $k$ is the size of the largest initial
clause~\cite{BenSassonWigderson}. The same type of trade-off was later
established for monomial size and degree for the Polynomial
Calculus~(PC) in \cite{ImpagliazzoPudlakSgall}, and for proof length
and rank for LS and LS$^+$~\cite{PitassiSegerlind2012}, i.e., the
proof systems that come out of the Lov\'asz-Schrijver LP and SDP
hierarchies~\cite{LovaszSchrijver}. To date, the question for PS and
SOS had remained open, and is answered here\footnote{Besides the
  proofs of the trade-off results for LS and LS$^+$, the conference
  version of \cite{PitassiSegerlind2012} claims the result for the
  stronger Sherali-Adams and Lasserre/SOS proof systems, but the claim
  is made without proof. The very last section of the journal version
  \cite{PitassiSegerlind2012} includes a sketch of a proof that,
  unfortunately, is an oversimplification of the LS/LS$^+$ argument
  that cannot be turned into a correct proof. The forthcoming
  discussion clarifies how our proof is based
  on, and generalizes, the one for LS/LS$^+$ in~\cite{PitassiSegerlind2012}.}.

Our proof of the trade-off theorem for PS follows the standard pattern
of such previous proofs with one new key ingredient. Suppose $Q$ is a
system of equations and inequalities that has a size $s$
refutation. Going back to the main idea from
\cite{CleggEdmondsImpagliazzo}, the argument for getting a degree $d$
refutation goes in four steps: (1) find a variable $x$ that appears in
many large monomials, (2) set it to a value $b\in\{0,1\}$ to kill all
monomials where it appears, (3) induct on the number of variables to
get refutations of $Q[x=b]$ and $Q[x=\bar{b}]$ which, if $s$ is small
enough, are of degrees $d-1$ and $d$, respectively, and (4) compose
these refutations together to get a degree $d$ refutation of $Q$. The
main difficulty in making this work for PS is step (4), for two
reasons.

The first difficulty is that, unlike Resolution and the other proof
systems, whose proofs are \emph{deductive}, the proofs of PS are
\emph{formal identities}, also known as \emph{static}. This means
that, for PS, the reasoning it takes to refute $Q$ from the degree
$d-1$ refutation of $Q[x=b]$ and the degree $d$ refutation of
$Q[x=\bar{b}]$ needs to be witnessed through a single polynomial
identity, without exceeding the bound $d$ on the degree. This is
challenging because the general simulation of a deductive proof by a
static one incurs a degree loss.  The second difficulty comes from the
fact that, for establishing this identity, one needs to use a duality
theorem that is not obviously available for degree-bounded PS
proofs. What is needed is a zero-gap duality theorem for PS proofs of
non-negativity that, in addition, holds tight at~\emph{each} fixed
degree $d$ of proofs. For SOS, the desired zero-gap duals are provided
by the levels of the Lasserre hierarchy. This was established in
\cite{JoszHenrion2015} under the sole assumption that the inequalities
include a ball contraint $B^2 - \sum_{i=1}^n x_i^2 \geq 0$ for some $B
\in \reals$. In the Boolean hypercube case, this can be assumed
without loss of generality. For PS, we are not aware of any published
result that establishes what we need, so we provide our own proof. At
any rate, one of our contributions is the observation that a zero-gap
duality theorem for PS-degree is a key tool for completing the step (4) 
in the proof of the trade-off theorem. We reached this
conclusion from trying to generalize the proofs for LS and LS$_+$ from
\cite{PitassiSegerlind2012} to SOS. In those proofs, the corresponding
zero-gap duality theorems are required only for the very special case
where $d=2$ and for deriving linear inequalities from linear
constraints. The fact that these hold goes back to the work of
Lov\'asz and Schrijver~\cite{LovaszSchrijver}.

In the end, the zero-gap duality theorem for PS-degree turned out to
follow from very general results in the theory of ordered vector
spaces. Using a result from \cite{PaulsenTomforde2009} that whenever a
pre-ordered vector space has an order-unit a zero-gap duality holds,
we are able to establish the following general fact: for any convex
cone $\mathscr{C}$ of provably non-negative polynomials and its
restriction $\mathscr{C}_{2d}$ to proofs of some even degree~$2d$, if
the the ball constraints $R - x^2 \geq 0$ belong to $\mathscr{C}_2$
for all variables $x$ and some $R \geq 0$, then a zero-gap duality
holds for $\mathscr{C}_{2d}$ in the sense that
$$\sup\{ r \in \reals : p-r \in \mathscr{C}_{2d}\} 
= \inf\{ E(p) : E \in \mathscr{E}_{2d} \},
$$ where $\mathscr{E}_{2d}$ is an appropriate dual space for
$\mathscr{C}_{2d}$. The conditions are easily seen to hold for
PS-degree and SOS-degree in the Boolean hypercube case, and we have
what we want. We use this in Section~\ref{sec:size-degree-trade-off},
where we prove the trade-off lemma, but defer its proof to
Section~\ref{sec:duality}.

In Section~\ref{sec:applications} we list some of the applications of
the size-degree trade-off for PS that follow from known degree lower
bounds. Among these we include exponential size SOS lower bounds for
Tseitin formulas, Knapsack formulas, and optimal integrality gaps for
sparse random instances of MAX-3-XOR and MAX-3-SAT. Except for
Knapsack formulas, for which size lower bounds follow from an easy
random restriction argument applied to the degree lower bounds
in~\cite{Grigoriev,GrigorievHirschPasechnik}, these size lower bounds
for SOS appear to be new.

%% file: section-2-preliminaries.tex
\section{Preliminaries} \label{sec:preliminaries}

For a natural number $n$ we use the notation $[n]$ for the set
$\{1,\ldots,n\}$. We write~$\reals_{\geq 0}$ and~$\reals_{>0}$ for the sets
of non-negative and positive reals, respectively and $\mathbb{N}$ for
the set of natural numbers. The natural logarithm is denoted $\log$,
and $\exp$ denotes base $e$ exponentiation.

\subsection{Polynomials and the Boolean ideal}

Let $x_1,\ldots,x_n$ and
$\bar{x}_1,\ldots,\bar{x}_n$ be two disjoint sets of variables. Each
$x_i,\bar{x}_i$ is called a pair of twin variables, where $x_i$ is the
basic variable and $\bar{x}_i$ is its twin. We consider polynomials
over the ring of polynomials with real coefficients and commuting
variables $\{x_i,\bar{x}_i : i \in [n]\}$, 
which we write simply as~$\mathbb{R}[x]$.  The
intention is that all the variables range over the Boolean domain
$\{0,1\}$, and that $\bar{x}_i = 1 - x_i$. Accordingly, let $I_n$ be
the Boolean ideal, i.e., the ideal of polynomials generated by the
following set of \emph{Boolean axioms} on the $n$ pairs of twin
variables:
\begin{equation}
B_n = \{ x_i^2 - x_i : i \in [n] \} \cup \{
\bar{x}_i^2 - \bar{x}_i : i \in [n] \} \cup \{ x_i + \bar{x}_i - 1 : i
\in [n] \} \label{eqn:booleanaxioms}
\end{equation}
We write $p \equiv q \mod I_n$ if $p-q$ is in $I_n$.

A monomial is a product of variables. A term is the product of a
non-zero real and a monomial. A polynomial is a sum of terms. For
$\alpha \in \mathbb{N}^{2n}$, we write $x^\alpha$ for the
monomial~$\prod_{i=1}^n x_i^{\alpha_i} \bar{x}_i^{\alpha_{n+i}}$, so
polynomials take the form $\sum_{\alpha \in I} a_\alpha x^\alpha$ for
some finite $I \subseteq \mathbb{N}^{2n}$. The~\emph{monomial size} of
a polynomial $p$ is the number of terms, and is denoted $\size(p)$. A
\emph{sum-of-squares polynomial} is a polynomial of the form $s =
\sum_{i=1}^k r_i^2$, where each $r_i$ is a polynomial in
$\mathbb{R}[x]$.
For a polynomial $p \in
\mathbb{R}[x]$ we write $\deg(p)$ for its degree. We think of
$\mathbb{R}[x]$ as an infinite dimensional vector space, and we write
$\mathbb{R}[x]_d$ for the subspace of polynomials of degree at most
$d$.

\subsection{Sums-of-Squares proofs}

Let $Q = \{q_1,\ldots,q_\ell, p_1,\ldots,p_m\}$ be an indexed set of
polynomials.  We think of the $q_j$ polynomials as inequality
constraints, and of the $p_j$ polynomials as equality constraints:
\begin{equation}
q_1 \geq 0,\ldots,q_\ell \geq 0 , \;\;\;
p_1 = 0,\ldots,p_m = 0.  \label{eqn:constraints}
\end{equation} 
Let $p$ be another polynomial. A
\emph{Sums-of-Squares (SOS) proof} of $p \geq 0$ from $Q$ is a formal
identity of the form
\begin{equation}\label{proof}
p = s_0 + 
\sum_{j \in [\ell]} s_j q_j +  
\sum_{j \in [m]} t_j p_j + 
\sum_{q \in B_n} u_q q,
\end{equation}
where $s_0$ and $s_1,\ldots,s_\ell$ are sums of squares of
polynomials, $s_j = \sum_{i=1}^{k_j} r_{i,j}^2$ for $j \in
[\ell]\cup\{0\}$, and $t_1,\ldots,t_m$ and all $u_q$ are arbitrary
polynomials.  The proof is of \emph{degree at most $d$} if $\deg(p)
\leq d$, $\deg(s_0) \leq d$, $\deg(s_j) + \deg(q_j) \leq d$ for each
$j \in [\ell]$, and $\deg(t_j) + \deg(p_j) \leq d$ for each $j \in
     [m]$. The proof is of \emph{monomial size at most~$s$} if
$$
\sum_{i=1}^{k_0} \size(r_{i,0}) + \sum_{j\in[\ell]}\sum_{i=1}^{k_j} \size(r_{i,j}) + \sum_{j
  \in [m]} \size(t_j) \leq s.
$$ 
This definition of size corresponds to the number of monomials of
an explicit SOS proof given in the form
$(s_0,s_1,\ldots,s_\ell,t_1,\ldots,t_m)$, where each $s_j$ is given in
the form $(r_{1,j},\ldots,r_{k_j,j})$, and all the $r_{i,j}$ and $t_j$
polynomials are represented as explicit sums of terms. Accordingly,
the monomials of the $r_{i,j}$'s and the $t_j$'s are called the
\emph{explicit monomials} of the proof.

Note that the $u_q$ polynomials are not considered in the definition
we have chosen of an explicit SOS proof, so they do not contribute to
its monomial size or its degree. The rationale for this is that
typically one thinks of the identity in~\eqref{proof} as an
equivalence
$$
p \equiv s_0 + \sum_{j \in [\ell]} s_j q_j + \sum_{j \in [m]} t_j p_j
\mod I_n
$$ and we want proof size and degree to not depend on how the
computations modulo the Boolean ideal~$I_n$ are performed.  For degree
this choice is further justified from the fact that one may always
assume that the degrees of the products $u_q q$ do not surpass the
degree $d$ in a proof of degree $d$. This follows from the fact that
$B_n$ is a Gr{\"o}bner basis for $I_n$ with respect to any monomial
ordering -- one can see this quite easily using Buchberger's Criterion
(see e.g.  \cite{idealsvarietiesangalgorithms}).  In particular upper
and lower bounds for the restricted definition of degree imply the
same upper and lower bounds for our liberal definition of degree, and
vice versa.  For monomial size, this goes only in one direction: lower
bounds on our liberal definition of monomial size translate into lower
bounds for a restricted definition of monomial size that
takes~$\sum_{q \in B_n} \size(u_q)$ also into account. Since our aim
is to prove lower bounds on the number of monomials in a proof,
proving our results for our more liberal definition of monomial size
makes our results only stronger.

%

\subsection{Positivstellensatz proofs}
 
This will be an extension of SOS. Let $Q = \{ q_1,\ldots,q_\ell,
p_1,\ldots,p_m\}$ be an indexed set of polynomials interpreted as
in~\eqref{eqn:constraints}. A \emph{Positivstellensatz proof} (PS) of
$p \geq 0$ from $Q$ is a formal identity of the form
\begin{equation}
p = s_{\emptyset} + \sum_{J \in \mathscr{J}} s_J \prod_{j \in J} q_j +
\sum_{j \in [m]} t_j p_j + \sum_{q \in B_n} u_q q, \label{proof-ps}
\end{equation} 
where $\mathscr{J}$ is a collection of non-empty subsets of $[\ell]$,
each $s_J$ is a sum-of-squares polynomial,~$s_J = \sum_{i=1}^{k_J}
r_{i,J}^2$, and each $t_j$ and $u_q$ is an arbitrary polynomial.  The
proof is of \emph{degree at most~$d$} if $\deg(p) \leq d$, $\deg(s_\emptyset)
\leq d$, $\deg(s_J) + \sum_{j \in J} \deg(q_j) \leq d$ for each $J \in
\mathscr{J}$, and~$\deg(t_j) + \deg(p_j) \leq d$ for each $j \in
        [m]$. The proof is of \emph{monomial size at most~$s$} if
$$
\sum_{i=1}^{k_0} \size(r_{i,\emptyset}) + \sum_{J\in
  \mathscr{J}}\sum_{i=1}^{k_J} \size(r_{i,J}) + \sum_{j \in [m]}
\size(t_j) \leq s.
$$ The proof has \emph{product-width} at most $w$ if each $J \in
\mathscr{J}$ has cardinality at most $w$. The~\emph{explicit
  monomials} of the proof are the monomials of the $r_{i,J}$'s and the
$t_j$'s.  It should be noted that~PS applied to a $Q$ that contains at
most one inequality constraint (i.e., $\ell\leq 1$) is literally
equivalent to SOS: any power of a single inequality is either a
square, or the lift of that inequality by a square.
 
As in SOS proofs, the definitions of monomial size and degree of a
proof do not take into account the $u_q$ polynomials.  Likewise, the
monomials in the products $\prod_{j \in J} q_j$ do not contribute to
the definition of monomial size. As above, this liberal definition
plays in favour of lower bounds in the case of monomial
size. For degree, ignoring the $u_q$'s does not really matter, again,
because $B_n$ is a Gr\"obner basis for~$I_n$.

\subsection{More on the definition of monomial size}

Starting at \cite{CleggEdmondsImpagliazzo,Alekhnovichetal}, counting
monomials in algebraic proof systems such as the Polynomial Calculus
(PC) is a well-established practice in propositional proof
complexity. One motivation for it comes from the fact that PC with
twin variables, called~PCR in~\cite{Alekhnovichetal}, polynomially
simulates Resolution, and the natural transformation that is given by
the proof turns the clauses of the Resolution proof into
monomials. Another motivation comes from the fact that, in the area of
computational algebra, the performance of the Gr\"obner bases method
appears to depend significantly on how the polynomials are
represented. In this respect, the sum of monomials representation of
polynomials features among the first and most natural choices to be
used in practice. That said, for the natural static version of PC
called Nullstellensatz (NS)
\cite{BeameImpagliazzoKrajicekPitassiPudlak1996}, let alone for SOS
and PS, counting monomials does not appear to have such a
well-established tradition. Note that in the presence of twin
variables, SOS monomial size is known to polynomially simulate
Resolution (see Lemma~4.6 in \cite{AtseriasLauriaNordstrom2016}, where
this is proved with a slightly different definition of SOS and
monomial size from the one above; the difference is minor). It follows
that the first of the two motivations for counting monomials in PC
carries over to SOS, and hence to PS.

The original Beame et al.\ and Grigoriev-Vorobjov papers
\cite{BeameImpagliazzoKrajicekPitassiPudlak1996,GrigorievVorobjov}
where NS and PS were defined first, size is never considered, only
degree. The subsequent Grigoriev's papers
on~SOS~\cite{Grigoriev,Grigoriev2001} did not consider size either. To
the best of our knowledge, the first reference that defines a notion
of size for (the version of) PS proofs (with~$w=0$) appears to be
\cite{GrigorievHirschPasechnik}, where the size of a proof is defined
as ``the length of a reasonable bit representation of all
polynomials'' in the proof. The same paper proves lower bounds on the
``number of monomials'' of an SOS proof (see Lemma~9.1 in
\cite{GrigorievHirschPasechnik}) without being precise as to whether
it is counting monomials in the $r_{i,0}$ polynomials (in the notation
of~\eqref{proof}), or in the expansion of $s_0$ as a sum of
terms. Note, however, that $\size(s_0) \leq \sum_i \size(r_{i,0})^2$,
hence the difference between these two possibilities is not terribly
critical. As with the squares $s_j$, the definitions in
\cite{GrigorievHirschPasechnik} are not explicit as to whether the
monomials in the $t_j$ polynomials (in the notation of~\eqref{proof}
again) contribute to the monomial size by themselves, or whether one
is to take into account the expansions of the products $t_j
p_j$. Unlike ours, the definitions in \cite{GrigorievHirschPasechnik}
do not distinguish between the $u_q$ polynomials that multiply the
Boolean axioms and the rest. 

The difference between counting the monomials of the $s_j$ (or the
$r_{i,j}$) polynomials versus counting those in the expansions of the
products $s_j q_j$ and $t_j p_j$ is again not critical if one is
satisfied with a notion of size \emph{up to} a polynomial factor that
depends on the size of the input. If one is to care about such
refinements of monomial size that take into account polynomial
factors, then a natural size measure for, say,~$t_j p_j$ could well
be~$\size(t_j) +\size(p_j)$ or even $\size(t_j)\cdot \size(p_j)$,
instead of $\size(t_jp_j)$. Note that~$\size(t_j)\cdot \size(p_j)$
corresponds to the number of monomials that one would encounter while
expanding the product $t_j p_j$ in the naive way \emph{before} merging
terms with the same monomial, and in particular, before any potential
cancelling of terms occurs.  In \cite{AtseriasLauriaNordstrom2016},
the monomial size of (their slightly different version of)
Lasserre/SOS is defined in terms of the expanded summands, which in
the notation of~\eqref{proof}, would correspond to $\size(s_0) +
\sum_j \size(s_j q_j) + \sum_j\size(t_j p_j) + \sum_q \size(u_q
q)$. In \cite{LauriaNordstrom} the same convention for defining
monomial size is used but the last sum over $q$ is omitted since they
work mod~$I_n$ by default.  For PS proofs as in~\eqref{proof-ps} that
have large product-width $w$, whether we count the monomials in
the~$s_J$ polynomials or in the expansions of the products~$s_J
\prod_{j \in J} q_j$ could make a significant difference, i.e.,
exponential in $w$. If we think of the proof in~\eqref{proof-ps} as
given by the indexed sequences $(s_J : J \in \mathscr{J} \cup
\{\emptyset\})$ and $(t_j : j \in [m])$, then counting only the
monomials in the~$s_J$ polynomial, or even better in the $r_{i,J}$
polynomials, looks like~the~natural~choice.

%
%
%

%% file: section-3-degree-reduction-lemma.tex
\section{Size-Degree Trade-Off} \label{sec:size-degree-trade-off}

In this section we prove the following.

\begin{theorem} \label{thm:main} For every two natural numbers $n$ and
  $k$, every indexed set $Q$ of polynomials of degree at most $k$ with
  $n$ pairs of twin variables, and every two positive integers $s$
  and~$w$, if there is a PS refutation from $Q$ of and product-width
  at most $w$ and monomial size at most $s$, then there is a PS
  refutation from $Q$ of product-width at most $w$ and degree at
  most~$4\sqrt{2(n+1)\log(s)}+kw+4$.
\end{theorem}

An immediate consequence is a degree criterion for size lower bounds:

\begin{corollary}
Let $Q$ be an indexed set of polynomials of degree at most $k$ with
$n$ pairs of twin variables, and let $w$ be a positive integer.  If
$d$ is the minimum degree and $s$ is the minimum monomial size of PS
refutations from $Q$ of product-width at most $w$, and $d \geq kw+4$,
then~$s \geq \exp((d-kw-4)^2/(32(n+1)))$.
\end{corollary}

The proof of Theorem~\ref{thm:main} will follow the standard structure
of proofs for degree-reduction lemmas for other proof systems, except
for some complications in the \emph{unrestricting} lemmas. These
difficulties come from the fact that PS proofs are static. The main 
tool around these difficulties is a tight Duality Theorem for 
degree-bounded proofs with respect to so-called \emph{cut-off functions} 
as defined next.

\subsection{Duality modulo cut-off functions}

Let $Q = \{ q_1,\ldots,q_\ell, p_1,\ldots,p_m \}$ be an indexed set of
polynomials interpreted as constraints as in~\eqref{eqn:constraints}.
A \emph{cut-off function} for $Q$ is a function $c : \Pow([\ell])
\disjointunion [m] \rightarrow \nats$ with~$c(J)~\geq~\sum_{j \in J}
\deg(q_j)$ for each $J \subseteq [\ell]$, and $c(j) \geq \deg(p_j)$
for each $j \in [m]$.  A PS proof as in~\eqref{proof-ps} has
\emph{degree mod $c$} at most $d$ if $\deg(p) \leq d$, $\deg(s_0) \leq
d$, $\deg(s_J) \leq d - c(J)$ for each $J \in \mathscr{J}$, and
$\deg(t_j) \leq d - c(j)$ for each $j \in [m]$.

Let $\mathrm{PS}^c_{w,d}(Q)$ denote the set of all polynomials $q$ of
degree at most $d$ such that $q \geq 0$ has a PS proof from $Q$ of
degree mod $c$ at most $d$ and product-width at most $w$.  We
write~$Q\vdash_{w,d}^c q\geq p$ if $q- p \in \mathrm{PS}^c_{w,d}(Q)$.
A \emph{pseudo-expectation} for $Q$ of degree mod $c$ at most~$d$ and
product-width at most $w$ is a linear functional $E$ from the of all
polynomials of degree at most $d$ such that $E(1) = 1$ and $E(q) \geq
0$ for all $q \in \mathrm{PS}^c_{w,d}(Q)$. We denote
by~$\mathscr{E}_{w,d}^c(Q)$ the set of pseudo-expectations for the
indicated parameters.

\begin{theorem}\label{thm:duality-for-sos} Let $d$ be a positive
  integer, let $Q$ be an indexed set of polynomials, let $c$ be a
  cut-off function for $Q$, let $w$ be a positive integer, and let $p$
  be a polynomial of degree at most~$2d$.  Then
$$
\sup\{ r \in \reals : Q \vdash_{w,2d}^c p \geq r \} =
\inf\{ E(p) : E \in \mathscr{E}^c_{w,2d}(Q) \}.
$$ 
Moreover, if the set $\mathscr{E}^c_{w,2d}(Q)$ is non-empty, then there is a
pseudo-expectation achieving the infimum; i.e., $\min\{ E(p) : E \in
\mathscr{E}^c_{w,2d}(Q) \}$ is well-defined.
\end{theorem}

Note that the statement of Theorem~\ref{thm:duality-for-sos} applies
only to even degrees. This comes as an artifact of the proof but is in
no way a severe restriction for the applications that we have in mind.
The definitions of degree for SOS and PS proofs as defined in
Section~\ref{sec:preliminaries} are special cases of the definitions
above for appropriate choices of $w$ and $c$.  Thus,
Theorem~\ref{thm:duality-for-sos} gives Duality Theorems for them. The
role of the cut-off function $c$ in our application below will be
explained in due time; i.e., after its use in the unrestricting
Lemma~\ref{lem:unrestricting} below. It is important for the lemmas
that follow that these duality theorems are tight in two ways: that
they have zero duality gap \emph{and} that they respect the degree;
i.e., the degree bound is the same for proofs and
pseudo-expectations. We defer the proof of
Theorem~\ref{thm:duality-for-sos} to Section~\ref{sec:duality} where a
more general statement is proved.

\subsection{Unrestricting lemmas}

For this section, fix three positive integers $n$, $d$ and $w$ for the
numbers of pairs of twin variables, degree, and product width. We also
fix an indexed set $Q = \{ q_1,\ldots,q_\ell, p_1,\ldots,p_m \}$ of
polynomials on the $n$ pairs of twin variables, and a cut-off function
$c$ for $Q$.

\begin{lemma} \label{lem:equality} Let $p$ and $q$ be polynomials of
  degree at most $2d$. If $p \equiv q \mod I_n$, then $E(p) =
  E(q)$ for any $E \in \mathscr{E}^{c}_{w,2d}(Q)$.
\end{lemma}

\begin{proof}
  The assumption that $p \equiv q \mod I_n$ implies that both $p-q$
  and $q-p$ belong to~$\mathrm{PS}^c_{w,2d}(Q)$. Hence $E(p) = E(q)$
  for any $E \in \mathscr{E}^c_{w,2d}(Q)$.
\end{proof}

\begin{lemma} \label{lem:liftingpseudos} 
  Let $x$ be one of the $2n$ variables and let $m$ be a monomial of
  degree at most~$2d-1$. Then $E(x) = 0$ implies $E(xm) = 0$ for any
  $E \in \mathscr{E}^c_{w,2d}(Q)$.
\end{lemma}

\begin{proof}
Let $m_1$ and $m_2$ be two monomials of degree at most $d-1$ and $d$,
respectively, such that $m = m_1m_2$. Note first that $E((xm_1)^2) =
0$, since $x - (xm_1)^2 \equiv (x - xm_1)^2 \mod I_n$ and all degrees
are at most $2d$. Hence, $0 = E(x) \geq E((xm_1)^2)\geq 0$ by
Lemma~\ref{lem:equality}. Let then~$a = E(m_2^2)$ and note that $a
\geq 0$. For every positive integer $k$ we have
\begin{align*}
E(xm) & \leq \frac{1}{2k}(E(2kxm_1m_2) + E((kxm_1 - m_2)^2)) = \frac{a}{2k}, \\
E(xm) & \geq \frac{1}{2k}(E(2kxm_1m_2) - E((kxm_1 + m_2)^2)) = -\frac{a}{2k},
\end{align*}
where in both cases the equalities follow from $E((xm_1)^2) = 0$ and
$E(m_2^2)=a$. Since $a \geq 0$ and the inequalities hold for every $k
> 0$ it must be that $E(xm)=0$ and the lemma is proved.
\end{proof}

For $q$ a polynomial on the $n$ pairs of twin variables, $i \in [n]$
an index, and $b \in \{0,1\}$ a Boolean value, we denote by $q[i/b]$
the polynomial that results from assigning $x_i$ to $b$ and
$\bar{x}_i$ to $1-b$ in $q$. We extend the notation to indexed sets of
such polynomials through~$Q[i/b]$ to mean $\{ q_j[i/b] : j \in [\ell]
\} \cup \{ p_j[i/b] : j \in [m] \}$. Note that $q_j[i/b]$ and
$p_j[i/b]$ are polynomials on~$n-1$ pairs of twin variables, and their
degrees are at most those of $q_j$ and $p_j$, respectively.


\begin{lemma} \label{lem:unrestricting}
  Let $i \in [n]$, let $Q_0$ and $Q_1$ be the extensions of $Q$ with
  the polynomials $p_{m+1} = x_i$ and $p_{m+1} = \bar{x}_i$,
  respectively, and let $c'$ be the extension of $c$ that maps $m+1$
  to $1$.  The following hold:
\begin{enumerate} \itemsep=0pt
\item[()] The function $c'$ is a cut-off function for both $Q_0$ and
  $Q_1$,
\item[(i)] If $Q[i/0] \vdash^c_{w,2d} -1 \geq 0$, then $Q_0
  \vdash^{c'}_{w,2d} -1 \geq 0$.
\item[(ii)] If $Q[i/1] \vdash^c_{w,2d} -1 \geq 0$, then $Q_1
  \vdash^{c'}_{w,2d} -1 \geq 0$.
\end{enumerate}
\end{lemma}

\begin{proof}
  \emph{()} is obvious. By symmetry we prove only \emph{(i)}.  Suppose
  that $Q[i/0] \vdash^c_{w,2d} -1 \geq 0$, say:
\begin{equation}
-1 = s_0 + \sum_{J \in \mathscr{J}} s_J \prod_{j \in J} q_j[i/0] + 
\sum_{j \in [m]} t_j p_j[i/0] +
\sum_{q \in B_n} t_q q[i/0].
\label{eqn:initial}
\end{equation}
For $j \in [\ell]$, write $q_j = \sum_{\alpha \in I_j} a_{j,\alpha}
x^\alpha$, let $J_j = \{ \alpha \in I_j : \alpha_i \geq 1 \}$ and
$K_j = \{ \alpha \in I_j : \alpha_i = 0 \text{ and } \alpha_{n+i} \geq 1\}$
and note that
\begin{align*}
q_j[i/0] = q_j + \sum_{\alpha \in J_j} a_{j,\alpha} (x^{\alpha}/x_i^{\alpha_i})
(-x_i^{\alpha_i}) + 
\sum_{\alpha \in K_j} a_{j,\alpha} (x^{\alpha}/\bar{x}_i^{\alpha_{n+i}})(1-\bar{x}_i^{\alpha_{n+i}}).
\end{align*}
Therefore $q_j[i/0] \equiv q_j + r_j x_i \mod I_n$
where
$$
r_j = {\sum_{\alpha \in K_j} a_{j,\alpha}(x^\alpha/\bar{x}_i^{\alpha_{n+i}}) - \sum_{\alpha \in J_j} a_{j,\alpha} (x^{\alpha}/x_i^{\alpha_i})}.
$$
Note that $\deg(r_j) \leq \deg(q_j)-1$ since $\alpha_i \geq 1$ for
$\alpha \in J_j$ and $\alpha_{n+i} \geq 1$ for $\alpha \in K_j$. Now
\begin{align}
s_J \prod_{j \in J} q_j[i/0] & \equiv s_J \prod_{j \in J} (q_j +r_j x_i) \mod I_n \\
& \equiv s_J \prod_{j \in J} q_j + \Big({\sum_{T \subseteq J \atop T \not= J}
s_J \prod_{j \in T} q_j \prod_{j \in J \setminus T} r_j}\Big) x_i \mod I_n.
\label{eqn:lastline}
\end{align}
Because $c$ is a cut-off function for $Q$
and $c'(J)=c(J)$, we have $\deg(s_J) \leq 2d-c(J) = 2d-c'(J)$. 
Likewise
for every $T \not= J$, we have:
\begin{align*}
\deg\Big({s_J \prod_{j \in T} q_j \prod_{j \in J\setminus T} r_j}\Big) & \leq
\deg(s_J)+\sum_{j \in T} \deg(q_j) + \sum_{j \in J\setminus T} \deg(r_j) \\
& \leq 2d-c(J) + \sum_{j \in J} \deg(q_j) - 1 \leq 2d-1 = 2d-c'(m+1).
\end{align*}
The second inequality follows from the facts that $J\setminus T \not=
\emptyset$ and $\deg(r_j) \leq \deg(q_j)-1$ for all $j \in [m]$, the
third inequality follows from the fact that $c$ is a cut-off function
for $Q$, and the equality follows from the definition of
$c'$. Hence, $Q_0 \vdash^{c'}_{w,2d} s_J \prod_{j \in J} q_j[i/0]$.  A
similar and easier argument with $t_j$ and $p_j$ in place of $s_J$ and
$\prod_{j \in J} q_j$ shows that $Q_0 \vdash^{c'}_{w,2d}
t_jp_j[i/0]$. This gives proofs for all terms in the right-hand side
of~\eqref{eqn:initial}, and the proof of the lemma is complete.
\end{proof}

Some comments are in order about the role of the cut-off function in
the above proof. First note that, at the semantic level, the
constraint $q_j[i/0] \geq 0$ is equivalent to the pair of constraints
$q_j \geq 0$ and $x_i = 0$. At the level of syntatic proofs, though,
these two representations of the same constraint behave differently:
although a lift $s_j q_j[i/0]$ of the restriction~$q_j[i/0] \equiv q_j
+ r_j x_i$ of $q_j$ may have its degree bounded by $2d$, the degree of 
its direct simulation through $s_j q_j + s_j r_j x_i$ could exceed $2d$. 
The role of the cut-off function is to restrict the lifts $s_j q_j[i/0]$ 
in such a way that their simulation through $s_j q_j + s_j r_j x_i$ remains 
a valid lift of degree at most $2d$; this is the case if, indeed, the 
allowed lifts $s_j q_j[i/0]$ of $q_j[i/0]$ are those satisfying 
$\deg(s_j) \leq 2d-c(j)$, where $c(j) \geq \deg(q_j)$. This is why 
$c$ is designed to depend only on the index $j$ (or $J$) and not on 
the polynomial indexed by $j$ (or~$J$).

\begin{lemma}\label{pseudoexpectations-positive-less-than-one} 
  Let $i \in [n]$, let $Q_0$ and $Q_1$ be the extensions of $Q$ with
  the polynomials $p_{m+1} = x_i$ and $p_{m+1} = \bar{x}_i$,
  respectively, and let $c'$ be the extension of $c$ that maps $m+1$
  to $1$.  The following hold:
\begin{itemize} \itemsep=0pt
\item[()] The function $c'$ is a cut-off function for both $Q_0$ and
  $Q_1$.
\item[(i)] If $Q_0 \vdash^{c'}_{w,2d} -1\geq 0$, then $E(x_i)
 > 0$ for any $E\in\mathscr{E}^c_{w,2d}(Q)$.
\item[(ii)] If $Q_1 \vdash^{c'}_{w,2d} -1\geq 0$, then $E(\bar{x}_i
) > 0$ for any $E\in\mathscr{E}^c_{w,2d}(Q)$.
\end{itemize}
\end{lemma}

\begin{proof}
  \emph{()} is obvious. We prove \emph{(i)}; the proof of \emph{(ii)}
  is symmetric. Suppose towards a contradiction that there is
  $E\in\mathscr{E}^c_{w,2d}(Q)$ such that $E(x_i) = 0$. We want to show
  that $E$ is also in
  $\mathscr{E}^{c'}_{w,2d}(Q_0)$. This contradicts the assumption that 
  $Q_0 \vdash^{c'}_{w,2d} -1\geq 0$. Let
\begin{equation}
s_\emptyset + \sum_{J\in \mathscr{J}} s_J \prod_{j \in J} q_j + 
\sum_{j \in [m]} t_j p_j + 
t_{m+1} x_i + \sum_{q \in B_n} t_q q
\label{eqn:expression}
\end{equation}
be a proof from $Q_0$ of degree mod $c'$ at most $2d$ and
product-width at most $w$. First note that $\deg(t_{m+1}) \leq
2d-c'(m+1) \leq 2d-1$. Therefore, Lemma~\ref{lem:liftingpseudos}
applies to all the monomials of $t_{m+1}$, so $E(t_{m+1}x_i) = 0$. The
rest of~\eqref{eqn:expression} will get a non-negative value
through~$E$, since by assumption $E$ is in $\mathscr{E}^c_{w,2d}(Q)$
and $c$ is $c'$ restricted to $\Pow([\ell]) \disjointunion [m]$. Thus,
$E$ is in $\mathscr{E}^{c'}_{w,2d}(Q_0)$.
\end{proof}

\begin{lemma} \label{lem:composing}
Let $i \in [n]$ and assume that $d \geq 2$. The following hold:
\begin{itemize} \itemsep=0pt
\item[(i)] If $Q[i/0]\vdash^c_{w,2d-2}-1\geq 0$ and $Q[i/1]\vdash^c_{2
d}-1\geq 0$, then $Q\vdash^c_{w,2d}-1\geq 0$.
\item[(ii)] If $Q[i/0]\vdash^c_{w,2d}-1\geq 0$ and $Q[i/1]\vdash^c_{2d
-2}-1\geq 0$, then $Q\vdash^c_{w,2d}-1\geq 0.$
\end{itemize}
\end{lemma}

\begin{proof}
Since in this proof $c$ and $w$ remain fixed, we write $\vdash_{2d}$
instead of $\vdash^{c}_{w,2d}$ and $\mathscr{E}_{2d}(Q)$ instead of
$\mathscr{E}^c_{w,2d}(Q)$, and act similarly for degree $2d-2$.  First
note that $-\bar{x_i}x_i = (x_i^2-x_i) - x_i(x_i + \bar{x}_i - 1)$,
and $d \geq 1$, so
\begin{equation}
\vdash_{2d} -\bar{x}_ix_i \geq 0. \label{eqn:pim}
\end{equation}
We prove \emph{(i)}; the proof of \emph{(ii)} is entirely
analogous. 

Assume $Q[i/0] \vdash_{2d-2} -1 \geq 0$.  By
Lemmas~\ref{lem:unrestricting}
and~\ref{pseudoexpectations-positive-less-than-one} and $d \geq 2$ we
have $E(x_i) > 0$ for any~$E \in \mathscr{E}_{2d-2}(Q)$. Then, by the
Duality Theorem, there exist $\epsilon > 0$ such that~$Q \vdash_{2d-2}
x_i \geq \epsilon$. To see this, let $\gamma = \sup\{ r \in \reals : Q
\vdash_{2d-2} x_i \geq r \} = \inf\{ E(x_i) : E \in
\mathscr{E}_{2d-2}(Q)\}$. If $\mathscr{E}_{2d-2}(Q)$ is empty, then
$\gamma = +\infty$ and any~$\epsilon > 0$ serves the purpose. If
$\mathscr{E}_{2d-2}(Q)$ is non-empty, then the Duality Theorem says
that the infimum is achieved, hence~$\gamma = E(x_i) > 0$ for some~$E$
in~$\mathscr{E}_{2d-2}(Q)$, and $\epsilon = \gamma/2 > 0$ serves the
purpose. Using $d \geq 2$ again, $Q\vdash_{2d}\bar{x}_i^2x_i\geq
\bar{x}_i^2\epsilon$, so
\begin{equation}
Q\vdash_{2d}\bar{x}_ix_i\geq \bar{x}_i\epsilon. \label{eqn:pam}
\end{equation}
Assume also $Q[i/1] \vdash_{2d} -1 \geq 0$. By
Lemmas~\ref{lem:unrestricting}
and~\ref{pseudoexpectations-positive-less-than-one} we have
$E(\bar{x}_i) > 0$ for any $E \in \mathscr{E}_{2d}(Q)$, and this time
$d \geq 1$ suffices.  By the same argument as before, by the Duality
Theorem there exist $\delta > 0$ such that $Q \vdash_{2d} \bar{x}_i
\geq \delta$. Now $d \geq 1$ suffices to get
\begin{equation}
Q \vdash_{2d} \bar{x}_i \epsilon \geq \delta\epsilon. \label{eqn:pum}
\end{equation}
Adding~\eqref{eqn:pim},~\eqref{eqn:pam} and~\eqref{eqn:pum} 
gives
$Q \vdash_{2d} 0 \geq \delta\epsilon$, i.e., $Q\vdash_{2d}
-1\geq 0$.
\end{proof}

\subsection{Inductive proof}

We need one more technical concept: a PS proof as in~\eqref{proof-ps}
is \emph{multilinear} if $s_0$ and $s_J$ are sums-of-squares of
multilinear polynomials for each $J \in\mathscr{J}$, and $t_j$ is a
multilinear polynomial for each $j \in [m]$.

\begin{lemma} \label{lem:implicitly} For every two positive integers
  $s$ and $w$ and every indexed set $Q$ of polynomials, if there is a
  PS refutation from $Q$ of monomial size at most $s$ and
  product-width at most $w$, then there is a multilinear PS refutation
  from $Q$ of monomial size at most $s$ and product-width at most $w$.
\end{lemma}

\begin{proof}
  Assume that $Q = \{q_1,\ldots,q_\ell,p_1,\ldots,p_m\}$ and that
  there is a refutation from $Q$ as in~\eqref{proof-ps}, with $s_0 =
  \sum_{i=1}^{k_0} r_{i,0}^2$ and $s_J = \sum_{i=1}^{k_J} r_{i,J}^2$
  for $J \in \mathscr{J}$, where the total number of monomials among
  the $r_{i,0}$, $r_{i,J}$ and $t_j$ is at most $s$. For each
  polynomial $r$ let $\overline{r}$ be its direct multilinearization;
  i.e., each power $x^{l}$ with $l \geq 2$ that appears in $r$ is
  replaced by $x$. It is obvious that $r \equiv \overline{r} \mod I_n$
  and also $r^2 \equiv \overline{r}^2 \mod I_n$, where $n$ is the
  number of pairs of twin variables in $Q$. Moreover, the number of
  monomials in $\overline{r}$ does not exceed that of $r$.  Thus,
  setting $s'_0 = \sum_{i=1}^{k_0} \overline{r_{i,0}}^2$, $s'_J =
  \sum_{i=1}^{k_J} \overline{r_{i,J}}^2$ and $t'_j = \overline{t_j}$
  we get
\begin{equation}
  -1 \equiv s'_0 + \sum_{J \in \mathscr{J}} s'_J \prod_{j \in J} q_j + 
  \sum_{j \in [m]} t'_j p_j \mod I_n,
\end{equation}
It follows that $Q$ has a multilinear refutation of monomial size at
most $s$.
\end{proof}


Theorem~\ref{thm:main} will be a consequence of the
following lemma for a suitable choice of $d$ and $c$:

\begin{lemma} \label{lem:technical} For every natural number $n$,
  every indexed set $Q$ of polynomials with $n$ pairs of twin
  variables, every cut-off function $c$ for $Q$, every real $s
  \geq 1$ and every two positive integers~$w$ and $d$, if there is a
  multilinear PS refutation from $Q$ of product-width at most $w$ with
  at most $s$ many explicit monomials of degree at least $d$ (counted
  with multiplicity), then there is a PS refutation from $Q$ of
  product-width at most $w$ and degree mod $c$ at most~$2d'+2d''$
  where $d' = d + \lfloor{2(n+1)\log(s)/d}\rfloor$ and $d'' =
  \max\{1,\lceil{(\max c)/2}\rceil\}$.
\end{lemma}

\begin{proof} 
  The proof is an induction on $n$.  Let $Q$ be an indexed set of
  polynomials with $n$ pairs of twin variables, let $c$ be a cut-off
  function for $Q$, let $s \geq 1$ be a real, let $w$ and $d$ be
  positive integers, and let $\Pi$ be a multilinear refutation from
  $Q$ of product-width at most~$w$ and at most $s$ many explicit
  monomials of degree at least~$d$. For $n = 0$ the statement is true
  because~$2d'' \geq 2\lceil{(\max c)/2}\rceil \geq \max c$.  Assume now
  that $n \geq 1$. Let $t \leq s$ be the exact number of explicit
  monomials of degree at least $d$ in $\Pi$. The total number of
  variable occurrences in such monomials is at least $dt$. Therefore,
  there exists one among the~$2n$ variables that appears in at least
  $dt/2n$ of the explicit monomials of degree at least~$d$. Let~$i \in
  [n]$ be the index of such a variable, basic or twin. If it is basic,
  let $a = 0$. If it is twin, let $a = 1$. Our goal is to show that
  \begin{equation}
  Q[i/a] \vdash^c_{2d'+2d''-2} -1 \geq 0 \;\;\;\;\text{ and }\;\;\;\; 
  Q[i/1-a] \vdash^c_{2d'+2d''} -1 \geq 0, \label{eqn:twothings}
  \end{equation} 
  for $d'$ and $d''$ as stated in the lemma.
%
  If we achieve so, then $d'+d'' \geq 2$ because $d' \geq d \geq 1$
  and $d'' \geq 1$, so
  Lemma~\ref{lem:composing} applies on~\eqref{eqn:twothings} to give
  $Q \vdash_{2d'+2d''}^c -1 \geq 0$, which is what we are after.

  Consider $Q[i/a]$ first. This is a set of polynomials on $n-1$ pairs
  of twin variables, and~$\Pi[i/a]$ is a multilinear refutation from
  it of product-width at most $w$ that has at most~$s' := t(1-d/2n)$
  explicit monomials of degree at least $d$. Moreover $c$ is a cut-off
  function for it. We distinguish the cases $s' < 1$ and $s' \geq 1$.
  If $s' < 1$, then all explicit monomials in $\Pi[i/a]$ have degree
  at most $d-1$.  Since $2d''\geq \max c$, this refutation has degree
  mod~$c$ at most $2(d-1)+2d'' \leq 2d'+2d''-2$. This gives the first
  part of~\eqref{eqn:twothings}.  If $s' \geq 1$, then first note that
  $d < 2n$. Moreover, the induction hypothesis applied to $Q[i/a]$ and
  $s'$, and the same $c$, $d$ and $w$, gives that there is a
  refutation from $Q[i/a]$ of product-width at most $w$ and degree mod
  $c$ at most $2d_a+2d''$, where
\begin{equation}
d_a = d + \lfloor{2n\log(t(1-d/2n))/d}\rfloor \leq d + 
  \lfloor{2(n+1)\log(s)/d}\rfloor - 1.
\label{eqn:first}
\end{equation}
Here we used the inequality $\log(1+x) \leq x$ which holds true for
every real $x > -1$, and the fact that $d < 2n$. This
gives the first part of~\eqref{eqn:twothings} since $d_a \leq d'-1$.

Consider $Q[i/1-a]$ next. In this case, the best we can say is that
$c$ is still a cut-off function for it, and that $\Pi[i/1-a]$
is a multilinear refutation from it of product-width at most $w$, that
still has at most~$s$ many explicit monomials of degree at least
$d$. But $Q[i/1-a]$ has at most $n-1$ pairs of twin variables, so the
induction hypothesis applies to it. Applied to the same $c$, $s$, $d$
and $w$, it gives that there is a refutation from $Q[i/1-a]$ of degree
mod~$c$ at most $2d_{1-a}+2d''$, where
\begin{equation}
  d_{1-a} = d + \lfloor{2n\log(s)/d}\rfloor \leq d + \lfloor{2(n+1)\log(s)/d}\rfloor.
\label{eqn:second}
\end{equation}
This gives the second part
of~\eqref{eqn:twothings} since $d_{1-a} \leq d'$. The
proof is complete.
%
\end{proof}

\begin{proof}[Proof of Theorem~\ref{thm:main}]
  Assume that $Q$ has a refutation of product-width at most $w$ and
  monomial size at most $s$.  Applying Lemma~\ref{lem:implicitly} we
  get a multilinear refutation with at most $s$ many explicit
  monomials, and hence with at most $s$ many explicit monomials of
  degree at least $d_0$, for any $d_0$ of our choice. We choose
\begin{equation}
d_0 := \lfloor{\sqrt{2(n+1)\log(s)}}\rfloor+1.
\end{equation}
By assumption $s \geq 1$ and we chose $d_0$ in such a way that $d_0
\geq 1$. Thus, Lemma~\ref{lem:technical} applies to any cut-off
function $c$ for $Q$, in particular for the cut-off function that is
$kw$ everywhere. This gives a refutation of product-width at most
$w$ and degree mod $c$ at most $2d'+kw+2$ with
\begin{equation}
d' \leq d_0+2(n+1)\log(s)/d_0 \leq 2\sqrt{2(n+1)\log(s)}+1.
\end{equation}
Since a proof of product-width at most $w$ and degree mod $c$ at most
$2d'+kw+2$ is also a proof of standard degree at most $2d'+kw+2$, the
proof is complete.
\end{proof}

%% file: section-4-applications.tex
\section{Applications} \label{sec:applications}

The obvious targets for applications of Theorem~\ref{thm:main} are the
examples from the literature that are known to require linear degree
to refute. For some of them, such as the Knapsack, the size lower
bound that follows was already known. For some others, the application
of Theorem~\ref{thm:main} yields a new result.

A note is in order: all the examples below 
%
%
are either systems of polynomial
equations, i.e., $\ell = 0$, or have a single inequality,~i.e.,~$\ell
= 1$. For such systems of constraints, PS and~SOS are literally
equivalent. For this reason, our size lower bounds for them are stated
only for~SOS (stating them for PS would be accurate, but also
misleading). 
%

\subsection{Tseitin, Knapsack, and Random CSPs}

The first set of examples that come to mind are the Tseitin formulas:
If $G_n = (V,E)$ is an $n$-vertex graph from a family $\{ G_n : n \in
\mathbb{N} \}$ of constant degree regular expander graphs, then the
formula $\mathrm{TS}_n$ has one Boolean variable $x_e$ for each $e \in E$ and
one parity constraint~$\sum_{e : u \in e} x_e = 1 \mod 2$ for each $u
\in V$. Whenever the degree $d$ of the graphs is even, this is
unsatisfiable when $n$ is odd. In the encoding of the constraints
given by the system of polynomial equations $Q = \{ \prod_{e : u \in
  e} (1-2x_e) = -1 : u \in V \}$, the Tseitin formulas~$\mathrm{TS}_n$ were
shown to require degree $\Omega(n)$ to refute in PS in Corollary~1 from
\cite{Grigoriev2001}. Since the number of variables of $\mathrm{TS}_n$ is
$dn/2$, the constraints in $Q$ are equations of degree $d$, and $d$ is
a constant, Theorem~\ref{thm:main} gives:

\begin{corollary}
There exists $\epsilon \in \reals_{>0}$ such that for every sufficiently
large $n \in \mathbb{N}$, every SOS refutation of $\mathrm{TS}_n$ has monomial
size at least $2^{\epsilon n}$.
\end{corollary}

Among the semialgebraic proof systems in the literature, exponential
size lower bounds for Tseitin formulas were known before for a proof
system called static LS$_+$ in
\cite{GrigorievHirschPasechnik,Itsykson2007}. Up to at most doubling the
degree, this can be seen as the subsystem of SOS in which every square
$s_j$ is of the very special form
$$
s_j = \Big({\Big({\sum_{i \in
  [n]} a_ix_i + b}\Big) \prod_{i \in I} x_i \prod_{j \in J} (1-x_j)}\Big)^2.
$$

A second set of examples are
the Knapsack equations~$2x_1 + \cdots + 2x_n = k$, which are
unsatisfiable for odd integers $k$. We denote them $\mathrm{KS}_{n,k}$.
These are known to require degree~$\Omega(\min\{k,2n-k\})$ to refute
in SOS \cite{Grigoriev}. Since the number of variables is $n$ and
the degree is one, Theorem~\ref{thm:main} gives an exponential size
$2^{\Omega(n)}$ lower bound when $k = n$.  For this example, an
exponential size lower bound for SOS was also proved in Theorem~9.1
from \cite{GrigorievHirschPasechnik} when $k = \Theta(n)$, so this
result is not new. We state the precise relationship that the
degree-reduction theorem gives in terms of $n$ and~$k$, which yields
superpolynomial lower bounds for~$k = \omega(\sqrt{n\log n})$.

\begin{corollary}
There exist $\epsilon \in \reals_{>0}$ such that for every
sufficiently large $n \in \mathbb{N}$ and $k \in [n]$, every SOS
refutation of $\mathrm{KS}_{n,k}$ has monomial size at least
$2^{\epsilon k^2/n}$.
\end{corollary}

The third set of examples come from sparse random instances of
constraint satisfaction problems. As far as we know, monomial size
lower bounds for these examples do not follow from earlier published
work without using our result, so we give the details.

When $C$ is a clause with $k$ literals, say $x_{i_1} \vee \cdots \vee
x_{i_\ell} \vee \bar{x}_{i_{\ell+1}} \vee \cdots \vee \bar{x}_{i_k}$,
we write $p_C$ for the unique multilinear polynomial on the
variables $x_{i_1},\ldots,x_{i_k}$ of $C$ that evaluates to the same
truth-value as $C$ over Boolean assignments; concretely $p_C =
1-\prod_{j=1}^\ell (1-x_{i_j}) \prod_{j=\ell+1}^k x_{i_j}$. More
generally, if $C$ denotes a constraint on $k$ Boolean variables, we
write $p_C$ for the unique multilinear polynomial on the variables of
$C$ that represents $C$ over Boolean assignments; i.e., such that
$p_C(x) = 1$ if $x$ satisfies $C$, and $p_C(x)=0$ if $x$ falsifies
$C$, for any $x \in \{0,1\}^n$.

\begin{theorem}[see Theorem 12 in \cite{Schoenebeck2008}] 
\label{thm:schoenebeck}
For every $\delta \in \reals_{>0}$ there exist $c,\epsilon \in
\reals_{>0}$ such that, asymptotically almost surely as $n$ goes to
infinity, if $m = \lceil{cn}\rceil$ and $C_1,\ldots,C_m$ are random
3-XOR (resp. 3-SAT) constraints on $x_1,\ldots,x_n$ that are chosen
uniformly and independently at random, then there is a
degree-$\epsilon n$ SOS pseudo-expectation for the system of polynomial
equations $p_{C_1} = 1, \ldots p_{C_m} = 1$, and at the same time
every truth assignment for $x_1,\ldots,x_n$ satisfies at most a
$1/2+\delta$ fraction (resp. $7/8+\delta$) of the constraints
$C_1,\ldots,C_m$.
\end{theorem}

It should be noted that it is not immediately obvious, from just
reading the definitions, that the statement of Theorem~12
in~\cite{Schoenebeck2008} gives the pseudo-expectation as stated in
Theorem~\ref{thm:schoenebeck}. However, the proof of Theorem~12
in~\cite{Schoenebeck2008} is by now sufficiently well understood to
know that Theorem~\ref{thm:schoenebeck} holds true as stated. One way
of seeing this is by noting that the proof of Theorem~12 in
\cite{Schoenebeck2008} and the proof of the lower bound for the
Tseitin formulas in Corollary~1 of \cite{Grigoriev2001} are
essentially the same. In particular Theorem~12 in
\cite{Schoenebeck2008} holds true also for proving the existence of
SOS pseudo-expectations as stated in Theorem~\ref{thm:schoenebeck}.

As an immediate consequence we get:

\begin{corollary} \label{cor:3xor3sat}
There exist $c,\epsilon \in \reals_{>0}$ such that, asymptotically
almost surely as $n$ goes to infinity, if $m = \lceil{cn}\rceil$ and
$C_1,\ldots,C_m$ are random 3-XOR (resp. 3-SAT) constraints on
$x_1,\ldots,x_n$ that are chosen uniformly and independently at
random, then every SOS refutation of $p_{C_1}=1,\ldots,p_{C_m}=1$ has
monomial size at least $2^{\epsilon n}$.
\end{corollary}

It is often stated that Theorem~\ref{thm:schoenebeck} gives optimal
integrality gaps for the approximability of MAX-3-XOR and MAX-3-SAT by
linear degree SOS. Corollary~\ref{cor:3xor3sat} is its analogue for
subexponential size SOS. There is however a subtelty in that the
validity of the integrality gap statement could depend on the encoding
of the objective function.  The next section is devoted to clarify
this.

\subsection{MAX-CSPs}

An instance $\mathscr{I}$ of the Boolean MAX-CSP problem is a sequence
$C_1,\ldots,C_m$ of constraints on $n$ Boolean variables.  We are
asked to maximize the fraction of satisfied constraints.
If $p_j$ denotes the unique multilinear polynomial on the variables of
$C_j$ that represents $C_j$, then the \emph{optimal value} for an
instance $\mathscr{I}$ can be formulated as follows::
\begin{equation}
 \mathrm{opt}(\mathscr{I}) := \textstyle{\max_{x \in \{0,1\}^n} \;\;\;
{\frac{1}{m}} \sum_{j=1}^m p_j(x)}. 
\label{eqn:poly}
\end{equation}
We could ask for the least upper bound on~\eqref{eqn:poly} that can be
certified by an SOS proof of some given complexity $c$, i.e., monomial size at
most $s$, degree at most $2d$, etc. There are at least three
formulations of this question. Using
the notation $\vdash_c$ to denote SOS provability with complexity $c$,
the three formulations are:
\begin{align}
  & \mathrm{sos''}_c(\mathscr{I}) := \textstyle{ \inf \{ \gamma \in \reals : \;\; \vdash_c\; 
\frac{1}{m} \sum_{j=1}^m p_j(x) \leq \gamma \} }, 
\label{eqn:direct} \\
& 
\mathrm{sos'}_c(\mathscr{I}) := \textstyle{\inf \{ \gamma \in \reals : \;\; 
\{ p_j(x) = y_j : j \in [m] \} \;\vdash_c\; 
\frac{1}{m} \sum_{j=1}^m y_j \leq \gamma \} }, \label{eqn:withvars} \\
& 
\mathrm{sos}_c(\mathscr{I}) := \textstyle{\inf \{ \gamma \in \reals : \;\; 
\{ p_j(x) = y_j : j \in [m] \} \cup  
\{ \frac{1}{m} \sum_{j=1}^m y_j \geq \gamma \}
\;\vdash_c\; -1 \geq 0 \}
}. \label{eqn:withvarsref}
\end{align}
The first formulation asks directly for the least upper bound on the
objective function of~\eqref{eqn:poly} that can be certified in
complexity $c$. The second formulation is similar but stronger since
it allows $m$ additional Boolean variables $y_1,\ldots,y_m$, and their
twins. The third is the strongest of the three as it asks for the
least value that can be proved impossible. In addition, unlike the
other two, the set of hypotheses in~\eqref{eqn:withvarsref} mixes
equations and inequality constraints.  It should be obvious that (for
natural complexity measures) we have $\mathrm{sos}_c(\mathscr{I}) \leq
\mathrm{sos'}_c(\mathscr{I}) \leq \mathrm{sos''}_c(\mathscr{I})$ so
lower bounds on sos$_c$ imply lower bounds for the other two.

Theorem~\ref{thm:schoenebeck} gives, by itself, optimal integrality
gaps for MAX-3-XOR and MAX-3-SAT for linear degree SOS in the
$\mathrm{sos''}_c$ formulation, when $c$ denotes SOS-degree. However,
the degree lower bound that follows from this formulation does not let
us apply our main theorem; the statement is not about refutations, it
is about proving an inequality, so Theorem~\ref{thm:main} does not
apply. In the following we argue that Theorem~\ref{thm:schoenebeck}
also gives optimal integrality gaps in the $\mathrm{sos'}_c$ and
$\mathrm{sos}_c$ formulations of the problems. Since the
$\mathrm{sos}_c$ formulation \emph{is} about refutations, our main
theorem will apply.

We write $\alpha_c(\mathscr{I})$ for the supremum of 
the $\alpha \in [0,1]$ for which
\begin{equation}
  \alpha\cdot\mathrm{sos}_c(\mathscr{I}) \leq \mathrm{opt}(\mathscr{I}) \leq \mathrm{sos}_c(\mathscr{I}) \label{eqn:approximation}
\end{equation}
holds. If $\mathscr{C}$ is a class of instances, then we write
$\alpha^*_c(\mathscr{C}) := \inf \{ \alpha_c(\mathscr{I}) :
\mathscr{I} \in \mathscr{C} \}$; the \emph{sos$_c$-approximation
  factor} for $\mathscr{C}$. It is our goal to show that
Theorem~\ref{thm:schoenebeck} implies that, for SOS proofs of
sublinear degree, the sos$_c$-approximation factor of MAX-3-XOR is at
most $1/2$, and that of MAX-3-SAT is at most $7/8$. These are
optimal. This will follow from Theorem~\ref{thm:schoenebeck} and the
following general fact about pseudo-expectations that (pseudo-)satisfy
all the constraints:

\begin{lemma}
  Let $\mathscr{I}$ be a MAX-CSP instance with $n$ Boolean variables
  and $m$ constraints of arity at most $k$,
  represented by multilinear polynomials $p_1,\ldots,p_m$, and let~$Q
  = \{ p_j(x) = 1 : j \in [m] \}$ and $Q' = \{ p_j(x) = y_j : j \in
  [m] \} \cup \{ \frac{1}{m} \sum_{j=1}^m y_j \geq 1 \}$.  If there is
  a degree-$2dk$ SOS pseudo-expectation $E$ for $Q$, then there is a
  degree-$2d$ SOS pseudo-expectation~$E'$ for $Q'$.
\end{lemma}

\begin{proof}
  Let $\sigma$ be the substitution that sends $y_j$ to $p_j(x)$ and
  $\bar{y}_j$ to $1-p_j(x)$ for $j = 1,\ldots,m$.  For each polynomial
  $p$ on the $x$ and $y$ variables, define $E'(p) := E(p[\sigma])$,
  where $p[\sigma]$ denotes the result applying the substitution to
  $p$. The proof that this works relies on the fact that if $p$ and
  $q$ are polynomial in the $x$ and $y$ variables, then $(pq)[\sigma]
  = p[\sigma] q[\sigma]$, and $\deg((pq)[\sigma]) \leq
  \deg(p[\sigma]q[\sigma]) \leq 2k(\deg(p)+\deg(q))$. In particular,
  squares maps to squares by the substitution. It is obvious that each
  equation $p_j(x) = y_j$ lifts: $E'(t(p_j(x)-y_j)) =
  E(t[\sigma](p_j(x)-p_j(x))) = E(0) = 0$.  It is equaly obvious that
  the inequality $\frac{1}{m}\sum_{j=1}^m y_j - 1 \geq 0$ lifts:
  $E'({s({\frac{1}{m}\sum_{j=1}^m y_j - 1})}) = \frac{1}{m}
  \sum_{j=1}^m E(s[\sigma] (p_j(x) - 1)) \geq 0$.  This completes the
  proof of the lemma.
\end{proof}

%

Combining this with Theorem~\ref{thm:schoenebeck} and Theorem~\ref{thm:main}
we get:

\begin{corollary} \label{cor:maxcsp}
For every $\delta \in \reals_{>0}$, there exist $r,\epsilon \in
\reals_{>0}$ such that if $c$ denotes SOS monomial size at most
$2^{\epsilon n}$, where $n$ is the number of variables, then
$\alpha^*_c(\text{MAX-3-XOR}) \leq 1/2+\delta$ (resp.
$\alpha^*_c(\text{MAX-3-SAT}) \leq 7/8+\delta$), and the gap is
witnessed by an instance $\mathscr{I}$ with $m = \lceil{rn}\rceil$
many uniformly and independently chosen random constraints, for which
$\mathrm{sos}_c(\mathscr{I}) = 1$ and $\mathrm{opt}(\mathscr{I}) \leq
1/2+\delta$ (resp. $\mathrm{opt}(\mathscr{I}) \leq 7/8+\delta$),
asymptotically almost surely as $n$ goes to infinity.
\end{corollary}

%% file: section-5-duality.tex
\section{Duality} \label{sec:duality}

In this section we finally prove the stated Duality Theorem for PS in a more general setting.
We start by recalling some basic facts about ordered vector spaces from \cite{PaulsenTomforde2009}. 
We prove the results for pre-ordered vector spaces rather than ordered ones since the polynomial 
spaces we will apply the results to carry a natural pre-order. 

\subsection{Vector spaces with order unit}

A pre-ordered vector space is a pair $\langle V,\leq\rangle$, where $V$ is a real vector space and $\leq$ is a pre-order that respects vector addition and multiplication by a non-negative scalar, i.e. the following hold for all $p,q,p_1,p_2,q_1,q_2 \in V$ and $a \in \reals_{\geq 0}$:
	\begin{itemize} \itemsep=0pt
		\item[(i)] $p_1\leq q_1$ and $p_2\leq q_2$ only if $p_1 + p_2 \leq q_1 + q_2$;
		\item[(ii)] $p\leq q$ only if $ap\leq aq$.
	\end{itemize}
Pre-ordered vector spaces arise naturally from convex cones of real
vector spaces. If $C\subseteq V$ is a convex cone, then the relation
defined by $p\leq_Cq$ if $q-p\in C$ satisfies the above requirements.
An element $e\in V$ is an \emph{order unit} for $\langle
V,\leq\rangle$ if for any $p\in V$ there is some~$r\in\reals_{\geq 0}$ such
that $re \geq p$. 

For the rest of this section let $\langle V,\leq \rangle$ be a
pre-ordered vector space with an order unit $e$.

\begin{lemma}\label{basic-properties}
	The following hold.
		\begin{itemize} \itemsep=0pt
			\item[(i)] $e\geq 0$;
			\item[(ii)] For every $p \in V$ and $r_1,r_2
                          \in \reals$ with $r_1 \leq r_2$, if
                          $r_1e\geq p$, then $r_2e\geq p$.
			\item[(iii)] For every $p\in V$ there is
                          $r\in\reals_{\geq 0}$ such that $re\geq p\geq -re$;
			\item[(iv)] If $-e\geq 0$, then $p \geq 0$ for
                          every $p\in V$.
		\end{itemize}
\end{lemma}

\begin{proof}
(i) There is some $r\in\reals_{\geq 0}$ such that $re\geq -e$, i.e. $(r+1)e\geq 0$, and so $e\geq 0$. (ii)~Now~$r_2-r_1\geq 0$ and so $(r_2-r_1)e\geq 0$. Thus $(r_2-r_1)e + r_1e\geq p$, i.e. $r_2 e\geq p$.
(iii)~Let~$r_1$ be such that $r_1e\geq p$ and let $r_2$ be such that $r_2e\geq -p$, and let $r = \max\{r_1,r_2\}$. Now~$re\geq p\geq -re.$
(iv) Suppose $-e\geq 0$ and let $r\in\reals_{\geq 0}$ be such that $re\geq -p$. Now also~$-re\geq 0$ and so $0 \geq -p$, i.e. $p\geq 0$.
\end{proof}

Let $U$ be a subspace of $V$.  A linear functional $L : U \rightarrow
\reals$ is \emph{positive} if $u \geq 0$ implies~$L(u) \geq 0$ for all $u \in U$. Equivalently, $L$ is
positive if it is order-preserving, i.e., if~$u \leq v$
implies $L(u) \leq L(v)$ for all $u,v \in U$. A positive linear 
functional $L$ on $V$ is a \emph{pseudo-expectation}
if $L(e) = 1$. We denote the set of all pseudo-expectations of $V$ by
$\mathscr{E}(V)$.

Suppose $U$ contains the order unit and let $p\in V$. By
Lemma~\ref{basic-properties}.\emph{(iii)} the following two sets are
non-empty:
	\begin{align*}
	p\downarrow U & = \{v\in U\colon p\geq v\},\\
	p\uparrow U & = \{v\in U\colon v\geq p\}.
	\end{align*}
If $L$ is any positive linear functional that is defined on $U$,
then $d_p^L = \sup\{ L(v)\colon v\in p\downarrow U\}$ and $u_p^L
= \inf\{L(v)\colon v\in p\uparrow U\}$ are real numbers and $d_p^L\leq
u_p^L$. Note also that if~$p\in U$, then~$d_p^L = L(p) = u_p^L$.

\begin{lemma}\label{key-lemma}
Let $U$ be a subspace of $V$ containing the order unit $e$, and let
$L$ be a positive linear functional on $U$. Then for any $p\in
V\setminus U$ and for any $\gamma \in \reals$ satisfying $d_p^L\leq
\gamma\leq u_p^L$ there is a positive linear functional $L'$ that is
defined on $\mathrm{span}(\{p\}\cup U)$, that extends $L$, and such
that $L'(p) = \gamma$.
\end{lemma}

\begin{proof}
Every element of $\mathrm{span}(\{p\}\cup U)$ can be written uniquely in form $ap + v$, where $a\in\reals$ and $v\in U$. Define $L'$ by 
$$L'(ap + v) = a\gamma + L(v).$$ It is easy to check that $L'$ is
linear map. We show that $L'$ is positive by considering a few cases.

Case (i) $a=0$. If $ap + v\geq 0$ and $a = 0$, then $v\geq 0$ and $L'(ap + v) = L(v)\geq 0$.
Case~(ii)~$a > 0$. Suppose that $ap + v\geq 0$ and $a > 0$. Then $p\geq
-(v/a)$, and so $L(-(v/a))\leq \gamma$, i.e. $0\leq a\gamma + L(v)$.
Case (iii) $a < 0$. Suppose that $ap + v\geq 0$ and $a < 0$. Then $-a > 0$,
and so $-(v/a)\geq p$. Hence $\gamma\leq L(-(v/a))$, and so $0\leq a\gamma
+ L(v)$.
\end{proof}

Now we can prove the general duality theorem for pre-ordered
vector spaces that admit an order unit. For a more general version of
this result, see \cite{PaulsenTomforde2009}.

\begin{theorem}\label{duality-theorem}
For any $p\in V$ it holds that
$$\sup\{r\in \reals\colon p\geq re\} = \inf\{E(p)\colon
E\in\mathscr{E}(V)\}.$$ Moreover, if the set $\mathscr{E}(V)$ is non-empty,
then there is a pseudo-expectation achieving the infimum,
i.e., $\min\{E(p)\colon E\in\mathscr{E}(V)\}$ is well-defined.
\end{theorem}

\begin{proof}
The inequality from left to right is clear. For the inequality from
right to left we distinguish two cases: whether $-e\geq 0$ or
not. If $-e\geq 0$, then $\mathscr{E}(V) = \emptyset$, since
$-1\ngeq 0$, so $\inf\{E(p)\colon E\in\mathscr{E}(V)\} =
+\infty$. On the other hand $\sup\{r\in \reals\colon p\geq re\} =
+\infty$ by Lemma~\ref{basic-properties}.\emph{(iv)}, so the claim
follows. If $-e\not\geq 0$, then $re \geq 0$ implies $r \geq 0$, so
the map defined by $L_0(re) = r$ for all $r \in \reals$ is a
positive linear functional on $U_0 = \mathrm{span}(\{e\})$.  Note now that~$d_p^{L_0} = \sup\{ r \in \reals : p \geq re \}$, and so, to prove the theorem, it suffices to show that there is some pseudo-expectation $E$ extending $L_0$ such that $E(p) = d_p^{L_0}$. 

If $p\in U_0$, then $L_0(p) = d_p^{L_0}$. On the other hand if $p \not\in U_0$, then by Lemma~\ref{key-lemma}, there is a positive linear functional $L'$ extending $L_0$ on $\mathrm{span}(\{e,p\})$ such that $L'(p) = d_p^{L_0}$. Now consider the set $\mathscr{A}$ of all positive linear functionals $L$ that are defined on a subspace~$U\subseteq V$ containing both $e$ and $p$, and satisfy $L(e) = 1$ and $L(p) = d_p^{L_0}$. By the argument above~$\mathscr{A}\neq \emptyset$. On the other hand $\mathscr{A}$ is closed under unions of chains and so, by Zorn's lemma, there is some maximal $E\in\mathscr{A}$.
			
Now the domain of $E$ is the whole of $V$, since otherwise we could
extend $E$ by using Lemma~\ref{key-lemma}, contradicting the
maximality of $E$. Hence $E$ is the pseudo-expectation we are after.
\end{proof}

\subsection{Order units for semi-algebraic proof systems}

For the purposes of this section we define a more general notion of
Positivstellensatz proof that works modulo an arbitrary ideal $I$, not
only the Boolean ideal $I_n$.  Let $I$ be an ideal of the polynomial
space $\reals[x]$, and let $Q = \{q_1\geq 0,\ldots,q_\ell\geq 0, p_1 =
0,\ldots, p_m = 0 \}$ be a set of constraints.  A PS proof mod $I$ of
$p \geq 0$ from $Q$ is an identity (of $\reals[x]/I$) of the form
\begin{equation}
p\equiv s_\emptyset + \sum_{J\subseteq \mathscr{J}} s_J\prod_{j\in J}q_j +
\sum_{j\in[m]} t_j p_j\mod{I}, \label{eqn:psproofmodI}
\end{equation} where $\mathscr{J}$ is a collection
of non-empty subsets of $[\ell]$, each $s_J$ is a sum-of-squares
polynomial,~$s_J = \sum_{i=1}^{k_J} r_{i,J}^2$, and each $t_j$ is an
arbitrary polynomial. 

A cut-off function for $Q$ is a function $c : \Pow([\ell])
\disjointunion [m] \rightarrow \nats$ with $c(J) \geq
\sum_{j \in J} \deg(q_j)$ for each $J \subseteq [\ell]$, and $c(j)
\geq \deg(p_j)$ for each $j \in [m]$.  A PS proof as
in~\eqref{eqn:psproofmodI} has \emph{degree mod $c$} at most $d$ if
$\deg(p) \leq d$, $\deg(s_0) \leq d$, $\deg(s_J) \leq d - c(J)$ for
each $J \in \mathscr{J}$, and $\deg(t_j) \leq d - c(j)$ for each $j
\in [m]$. It has \emph{product-width} at most $w$ if each $J \in
\mathscr{J}$ has cardinality at most~$w$.  We write
$\mathrm{PS}^{c,I}_{w,d}(Q)$ for the convex cone of all polynomials
$p$ such that $p \geq 0$ has a PS proof mod $I$ from $Q$ of degree mod
$c$ at most $d$ and product-width at most $w$.  We will write~$Q
\vdash^{c,I}_{w,d} p \geq q$ if $p-q \in
\mathrm{PS}^{c,I}_{w,d}(Q)$, and denote by
$\mathscr{E}^{c,I}_{w,2d}(Q)$ the set of pseudo-expectations over the
pre-ordered vector space determined by this cone.
These definitions agree with those used in
Section~\ref{sec:size-degree-trade-off} when $I = I_n$.

We show that over any ideal $I$, any cut-off function $c$ and any
product-width $w$, if $Q$ proves that each variable is bounded in
degree two, then the constant polynomial $1$ is an order unit for
$Q$. We prove this in a series of lemmas. In order to simplify the
notation, for these lemmas we write $\vdash_d$ instead of
$\vdash^{c,I}_{w,d}$.

\begin{lemma}
If $Q\vdash_2 R\geq x^2$ for every variable $x$ for some
$R\in\reals_{\geq 0}$, then for any monomial $m$ of degree at most $d$ and
any $a\in\reals$ there is $b\in\reals_{\geq 0}$ such that
$$Q\vdash_{2d} am^2 + b\geq 0.$$
\end{lemma}
\begin{proof}
We prove the claim by induction on the degree of $m$. If $\deg(m) =
0$, then the claim is trivial.  Suppose then that $\deg(m) > 0$.  If
$a \geq 0$, then the claim is again clear: $am^2 =
(\sqrt{a}m)^2$. Suppose that $a < 0$ and let $x$ and $m_0$ be such
that $m = x m_0$.  By assumption $Q\vdash_2 R - x^2\geq 0$, and so $Q\vdash_{2d} (\sqrt{-a}m_0)^2(R-x^2) \geq 0$. By induction hypothesis applied to $m_0$ and $a$ there is
$b_0\in\reals_{\geq 0}$ such that $Q\vdash_{2d} a R m_0^2 + b_0\geq
0$. By adding we have that $Q\vdash_{2d} am^2
+ b_0 \geq 0$.
	\end{proof}
	
\begin{lemma} \label{lem:monolll}
If $Q\vdash_2 R\geq x^2$ for every variable $x$ for some
$R\in\reals_{\geq 0}$, then for any monomial $m$ of degree at most $2d$ and
any $a\in\reals$ there is $b\in\reals_{\geq 0}$ such that
$$Q\vdash_{2d} am + b\geq 0.$$
\end{lemma}
\begin{proof}
Let $m_0$ and $m_1$ be monomials of degree at most $d$ such that
$m = m_0m_1$. Now if $a\geq 0$, then $(\sqrt{a/2}m_0 +
\sqrt{a/2}m_1)^2 = (a/2)m_0^2 + am + (a/2)m_1^2$. Now, by previous
lemma, there are non-negative $b_0$ and $b_1$ such that
$Q\vdash_{2d} (-a/2)m_i^2 + b_i\geq 0$ for
$i\in\{0,1\}$. Hence $Q\vdash_{2d} am + b_0 + b_1 \geq 0$. If
$a < 0$, then $(\sqrt{-a/2}m_0 - \sqrt{-a/2}m_1)^2 = (-a/2)m_0^2 + am
+ (-a/2)m_1^2$. Now, again by previous lemma, there are non-negative $b_0$
and $b_1$ such that $Q\vdash_{2d} (a/2)m_i^2 + b_i\geq 0$ for
$i\in\{0,1\}$. Hence $Q\vdash_{2d} am + b_0 + b_1 \geq 0$.
	\end{proof}
	
\begin{lemma}
If $Q\vdash_2 R\geq x^2$ for every variable $x$ for some
$R\in\reals_{\geq 0}$, then for any polynomial $p$ of degree at most $2d$
there is $r \in \reals_{\geq 0}$ such that
$$Q\vdash_{2d} r\geq p.$$
\end{lemma}

\begin{proof} Immediate from Lemma~\ref{lem:monolll}. \end{proof}

This
establishes the existence of an order-unit hence, by
Theorem~\ref{duality-theorem}, we have:

\begin{corollary} \label{cor:dualityball}
Let $d$ be a positive integer, let $Q$ be an indexed set of
polynomials, let $c$ be a cut-off function for $Q$, let $w$ be a
positive integer, let $I$ be an ideal of $\reals[x]$, and let $p$ be a
polynomial of degree at most $2d$. If $Q \vdash^{c,I}_{w,2} R \geq
x^2$ for every variable $x$ for some $R \in \reals_{\geq 0}$, then
$$ 
\sup\{ r \in \reals : Q \vdash_{w,2d}^{c,I} p \geq r \} = \inf\{ E(p) :
E \in \mathscr{E}^{c,I}_{w,2d}(Q) \}.
$$
Moreover, if the set $\mathscr{E}^{c,I}_{w,2d}(Q)$ is non-empty, then there is a
pseudo-expectation achieving the infimum; i.e., $\min\{ E(p) : E \in
\mathscr{E}^{c,I}_{w,2d}(Q) \}$ is well-defined.
\end{corollary}

For the Boolean ideal $I_n$, the assumption that $Q \vdash^{c,I}_{w,2}
R \geq x^2$ holds for every variable~$x$ is fulfilled with $R = 1$
since $1-x^2 \equiv (1-x)^2 \mod I_n$. This gives
Theorem~\ref{thm:duality-for-sos}. In the~$\pm 1$ representation of
the Boolean hypercube, i.e., modulo the ideal $I'_n$ generated by the
axioms~$B'_n := \{ 1-x_i^2, 1-\bar{x_i}^2, x_i+\bar{x}_i : i \in [n]
\}$, the assumption is fulfilled also with $R=1$ since in this case
$1-x^2 \equiv 0 \mod I'_n$.

%% file: section-6-concluding-remarks.tex
\section{Concluding Remarks} \label{sec:concluding-remarks}

In this paper we addressed the question of size-degree trade-offs for
PS and SOS. Some questions remain open. Most importantly, is the
$O(\sqrt{n\log(s)}+kw)$ upper bound in the degree-reduction lemma
tight? For Resolution and PC, whose size-width/degree trade-offs adopt
the same form, the bound is known to be tight. In both cases the
Ordering Principle~(OP) witnesses the necessity of the square root of
the number of variables in the upper bound
\cite{BonetGalesi,LauriaGalesi}. In this respect, it should be noted
that it was recently shown that OP$_n$, which has $N = n^2$ variables,
can be refuted in degree $O(\sqrt{n})$, whence degree $O(\sqrt[4]{N})$, in
SOS \cite{Potechin2018}. Since the relationship between $N$ and
$\sqrt{n}$ is a $4$-th root, this means that OP$_n$ cannot be used for
witnessing the necessity of the square root of the number of variables
in our theorem. But can OP$_n$ be used to show that at least some
fixed root $\sqrt[r]{n}$ of $n$ is required? So far, the best SOS degree
lower bound for OP$_n$ known is superconstant \cite{Potechin2018}.

Although it looks unlikely that the dependence of
$O(\sqrt{n\log(s)}+kw)$ on the product-width $w$ could be improved by
refining the current method, it is not even known whether there are
examples that separate PS from SOS. Could PS collapse to SOS with
respect to size or degree? Related to this, a comment worth making is
that there is a general well-known technique for transforming
inequalities $P \geq 0$ into equalities $P - z^2 = 0$, where $z$ is a
fresh variable. This looks relevant since, in the absence of
inequalities, PS collapses to SOS just by definition. On the other
hand, note that the new variable $z$ that is introduced by this method
is not Boolean, which takes us outside the Boolean hypercube.